\newcommand*{\LONG}{}%
\pdfoutput=1

\ifdefined\LONG
\documentclass[11pt]{article}
\usepackage{amsfonts,amsmath,amssymb,latexsym,comment,amsthm}

\usepackage[letterpaper]{geometry}
\usepackage{fullpage}

\usepackage{mathrsfs}
\usepackage{times}

\else 
\documentclass{sig-alternate-2013}
\usepackage{amsfonts,amsmath,amssymb,latexsym,comment }
\usepackage{fullpage}
\fi

\usepackage{qtree}
\usepackage{epsfig}
\usepackage{graphicx}
\usepackage{subfigure}

\usepackage{xspace}
\usepackage{cases}
\ifx\pdftexversion\undefined
\usepackage[colorlinks,linkcolor=black,filecolor=black,citecolor=black,urlcolor=black,pdfstartview=FitH]{hyperref}
\else
 \usepackage[colorlinks,linkcolor=blue,filecolor=blue,citecolor=blue,urlcolor=blue,pdfstartview=FitH]{hyperref}
\fi
\usepackage[ruled,boxed,commentsnumbered,noresetcount]{algorithm2e}
\usepackage{enumitem}

\usepackage{pifont}

\def\squarebox#1{\hbox to #1{\hfill\vbox to #1{\vfill}}}

\newtheorem{theorem}{Theorem}

\newtheorem{lemma}{Lemma}

\newtheorem{definition}{Definition}

\newtheorem{claim}{Claim}
\newtheorem{corollary}{Corollary}

\newcommand{\namedref}[2]{\hyperref[#2]{#1~\ref*{#2}}}
\newcommand{\sectionref}[1]{\namedref{Section}{#1}}

\newcommand{\algref}[1]{\namedref{Algorithm}{#1}}

\newcommand{\lemmaref}[1]{\namedref{Lemma}{#1}}

\newcommand{\corollaryref}[1]{\namedref{Corollary}{#1}}

\newcommand{\lref}[1]{\namedref{Line}{#1}}

\def\tbasyn{{$\B$-Byzantine resilient asynchronous system}\xspace}
\def\toasyn{{$\B$-resilient asynchronous system}\xspace}
\def\ttbasyn{{$t$-Byzantine resilient asynchronous system}\xspace}
\def\ttoasyn{{$t$-resilient asynchronous system}\xspace}

\newcommand{\G}{\mathcal{G}}
\newcommand{\B}{\mathcal{B}}
\newcommand{\BIMObb}{\mbox{BI-MO$(B_{initial})$}\xspace}
\newcommand{\BISYNt}{\mbox{BI-Synch}\xspace}

\newcommand{\cosend}{\mbox{\sc co\_send}\xspace}

\newcommand{\ie}{\emph{i.e.,~}}

\newcommand{\tri}[1]{\left<#1 \right>}

\newcommand{\ang}[1]{\left\langle {#1} \right\rangle}

\def\beginsmall#1{\vspace{-\parskip}\begin{#1}\itemsep-\parskip}
\def\endsmall#1{\end{#1}\vspace{-\parskip}}

\newcommand{\F}{\mathcal{F}}

\newcounter{todocounter}
\newcommand{\todonum}
{\stepcounter{todocounter}{(\thetodocounter)}}

\newcommand{\dd}[1]{\textbf{\color{blue}
[[\todonum\ dd: #1]]}}

\newcommand{\tb}{\makebox[0.6cm]{}}
\newcommand{\due}{\makebox[1cm]{}}
\newcommand{\tre}{\makebox[1.3cm]{}}

\newcommand{\hide}[1]{}

\newcommand{\commentout}[1]{}

\newcommand{\M}{\mathcal{M}}

\newcommand{\I}{\mathcal{I}}

\newcommand{\ccore}{{\small\textsc{common\!\_core}}\xspace}

\begin{document}

\title{
Byzantine Processors and Cuckoo Birds:\\
Confining Maliciousness to the Outset}
\author{Danny Dolev (HUJI)\footnote{Danny Dolev is Incumbent of the Berthold Badler Chair in Computer Science. 
The Rachel and Selim Benin School of Computer Science and Engineering
Edmond J. Safra Campus. This work was supported in part by the HUJI Cyber Security Research Center in conjunction with the Israel National Cyber Bureau in the Prime Minister's Office.
This research project was supported in part by The Israeli Center of Research Excellence (I-CORE) program, (Center  No. 4/11).  This work was supported in part by the Ministry of Education of China through its grant to Tsinghua University.
Email: danny.dolev@mail.huji.ac.il.} \ and
Eli Gafni (UCLA)\footnote{
Eli Gafni was partially supported by Israel-USA BSF grant 20152316 and an NSF grant 20170416.
This work was supported in part by the Ministry of Education of China through its grant to Tsinghua University.
Email: gafnieli@gmail.com.}\\[1.2ex]
\vspace{-2em}
}

\date{}

\maketitle

\thispagestyle{empty}

\section*{Abstract}
\vspace{-1em}

Are there Byzantine Animals?
\\
A Fooling Behavior is exhibited by the Cuckoo bird. It sneakily replaces some of the eggs of other species with its own. Lest the Cuckoo extinct itself by destroying its host, it self-limits its power: It does not replace too large a fraction of the eggs.  
Here, we show that any Byzantine Behavior that does not destroy the system it attacks, i.e. allows the system 
to solve an easy task like $\epsilon$-agreement, then its maliciousness can be confined to
be the exact replica of the Cuckoo bird behavior: Undetectably replace an input of a processor and let the processor behave correctly thereafter with respect to the new input. In doing so we reduce the study of Byzantine behavior to fail-stop
(benign)  behavior with the Cuckoo caveat of a fraction of the inputs replaced.
This goes beyond the reductions shown in the past that apply to colorless tasks.
We establish a complete correspondence between the Byzantine and the Benign, modulo different thresholds, and replaced inputs.

This work is yet another step in a line of work unifying seemingly distinct distributed system models, dispelling the Myth that Distributed Computing is a plethora of distinct isolated models, each requiring its specialized tools and ideas in order to determine solvability of tasks (that is, Computability rather than Complexity).
Thus, hereafter, Byzantine Computability questions can be reduced to questions in the benign failure setting.
But vice versa too.
In the more structured settings of asynchronous benign failures and synchronous
Byzantine failures, researchers investigated correlated faults.
We show that the known results
about correlated faults in the asynchronous benign setting can be imported verbatim to the
asynchronous Byzantine setting. 
Finally, as in the benign case in which we have the property that a processor can output once its faulty behavior stops for long enough, we show this can be done in a similar manner in the Byzantine case. This necessitated
the generalization of Reliable Broadcast to what we term Recoverable Reliable Broadcast.

\newpage

\section{Introduction}
\vspace{-1em}


``Deep Learning'' systems are built. They work, and will probably be the underpinning of a new Industrial Revolution.
It happens now and it proceeds lacking serious Theoretical Foundation. Theory, if developed, will have to play catch-up, explaining why the systems work the way they do.
Distributed Computing history is not unlike the current state of Deep Learning. Distributed Systems have been built and changed the world.
Theory played and still plays catch-up. But does it play successfully? The hallmark of Theory is Unification. Reducing a seemingly chaos into few principles.
This paper, motivated and helped by other recent successful unification steps, shows
that Byzantine Models are, to paraphrase Ecclesiastes (Koheleth) \cite{Koheleth}, ``Nothing New Under the Sun.''

And, it shows the value of Theory. This paper shows that the
solvability of tasks in the Byzantine model can be reduced to
a question about the solvability of tasks in the benign failure model. For a designer of a distributed system dealing with Byzantine processors, knowing whether the assumptions made theoretically allow for a solution, is as important as for a designer of a centralized system to know whether the system 
actually solves an $NP$-complete problem, or worse, and Undecidable
one. The reduction we introduce simplifies answering this solvability question. 

A limited reduction that applies mainly to the so called \emph{ colorless} tasks
has been derived recently in \cite{MTH14} and then rediscovered in \cite{IRS16}. 
The authors of \cite{MTH14} claim their reduction
is applicable conceptually to more than colorless tasks.
In their later paper \cite{MHVG15} there is no mention of \cite{MTH14},  that with respect to the non-colorless task of Multi-Dimensional Epsilon-Agreement \cite{MH13,VG13}, a task that we know now can be reduced to benign failures according to our Cuckoo model, they could have treated it in the benign model instead of resorting to analysis in the Byzantine model, with all the clutter it carries.  The Cuckoo model, by doing away with the notion of ``byzantine processors'' 
removes barriers that are creating blind spots in our understanding. We give an explicit mapping 
of the Byzantine in terms of the Benign. The era of Byzantine Computability is over!

This investigation was motivated by recent unification steps that gave rise to our inquest as to their ramifications to the Byzantine model. 

In \cite{AG15}, it was shown that as far as solving tasks, asynchronous models
can be reduced to synchronous models with mobile-faults. If this is the case in the benign setting, what about the Byzantine? Also, there \cite{AG15} processors are always ``correct,'' they do not ``crash.'' Their seemingly faulty behavior is due to the communication subsystem that is faulty. If we view Byzantine behavior not as a pitch-fork equipped Lucifer which takes over the soul of a processor, but just takes over a processor's outgoing communication links, then there is meaning to what Byzantine processor outputs, it is a victim rather than the perpetrator. 

Of course the Byzantine Model \cite{Agree80,Agree82} came to capture a processor malfunction. The point is that a theoretical model should be faithful in producing results commensurate with observable behavior. It should not necessarily try to capture the precise mechanism by which this behavior occurs in reality. In the benign faults model we try to capture ``crash,'' yet the right way to look at it is just as a worry about crash that therefore can be captured by communication delay, and
thus allows a processor to rejoin after suffering the delay. Similarly, a Byzantine processor suffers malicious communication tampering, but once the tampering ceases the processor should be able to join in. Thus, we actually generalize the Byzantine Model of almost four decades, streamlining it with the benign fault mechanism modeling, while maintaining all legacy results.

\hide{
The reduction in \cite{AG15} extends another 
early unification step
of shared-memory and message passing,
\cite{ABD95}, beyond $t<n/2$, to any number of failures albeit at the cost of making the reduction
allow for processor's-starvation, i.e. the system will progress but individual processor's action might be indefinitely postponed. Starvation plays role in long-lived problems. 
The paper \cite{AG15} investigates the correspondence between message-passing and shared-memory with respect to one-shot problems, referred to as {\emph tasks}, as we do here between the Byzantine and the Benign failures.
}
\hide{
More importantly the idea in \cite{AG15} is that processors never ``crash.'' 
It is rather the communication substrate that drops messages. Thus, if processors are always ``correct'' what
happens when we carry this view to the Byzantine model? It then becomes meaningful to ask
all processors to output, if somehow they succeeded to ``announce'' an input. This is the key
idea which let us proceed beyond colorless-tasks.
\dd{don't we need a ref here?}
}

The other motivating recent result \cite{GG11,DFGK15} 
raised the idea, in the benign case, that the connection
between a processor and the thread it is supposed to execute, is more tenuous than was thought.
The connection can
be confined to the processor providing an input to the thread, and then fetching the output, obtained by the thread termination,
to be be delivered to the user.
Between providing input and delivering an output to the user, advancing the program-counters of active-threads - those whose input has been
provided and whose execution has not been terminated, is a communal effort.
It is a negotiation between all processors what a specific thread has ``read'' as the system state.

What about the Byzantine model then? If advancing program-counters is a communal effort, what meaning
there is to a ``misbehaving-thread?''  All processors are responsible for it. Not surprisingly, the Cuckoo model we derive confines Lucifer only to the point of telling the system 
what input some of the threads should be run with. Although motivated by \cite{GG11,DFGK15}
we end up associating a processor with executing its thread and the communal effort is about accepting or rejecting a proposed step of a processor. It turned out to be less involved than our initial ``pure'' communal effort.

In an initial effort (\cite{DG16}) we got to provide reduction that applies for the ``correct'' processors, i.e. those whose communication links were never taken over by a Byzantine adversary. But how do we proceed once those correct processors output, if we now assume that the Byzantine Adversary gradually relinquishes its hold of some of the processors it interfered with before. How can a processor made whole (albeit with respect to its possible fake input) after other processors might have recognized it is faulty and stopped listening to it? 

Current classical technique of ``Reliable Broadcast'' (\cite{Bracha:1984})
does not allow recovery from such an implicit shut-off. Yet, in the benign failure ``recovery'' is trivial, as a seemingly ``crashed'' processor comes alive. Can it be that this fiat is doomed in the Byzantine case? This will contradict our sense of elegance that there should be \emph{complete} analogy between the 
Byzantine and the Benign systems. Indeed, the most time-consuming part of the technical research involved in this paper was to eventually abandon the belief and the trials at a proof that this incomplete analogy is inherent and actually come with a modification to Reliable Broadcast (Recoverable Reliable Broadcast) that makes the analogy complete.

It also necessitated our generalization of the notion of what a model is while leaving it compatible with the past definitions.
Traditionally, a model prescribes progress conditions independent of the execution
of an algorithm. Here we allow progress conditions to change as a function
of the outputs delivered. We assume a model in which once $n-t$ outputs  are delivered then at least from there on
one additional processor that has not outputted will not experience further attacks on its communication anymore.
\hide{
It would have been nice to have this condition as an eventual one, not requiring it to happen immediately after outputs are
delivered. It is an open question whether recovery under the eventual condition is possible.
}

Finally, if our reduction is complete it should go the other way around too. We should be able to import models and results from the Benign to the Byzantine. In the
synchronous and then the asynchronous benign models
a body of work analyzed Correlated Faults 
(Also called Adversary Model) \cite{HirtM97,AFM99, DFGT11, HR10, GK11}, and characterized the computability of such models. Our reduction allows to do same for the  Asynchronous Byzantine once we have derived Reliable Broadcast \cite{{Bracha:1984}} for Correlated Faults. We do just that.
That said, we nevertheless keep the paper cast in the $t$-resilience and faulty/correct processors nomenclature, as after almost four decades most of us are more comfortable with it, even though it is exactly what this paper advocates to forget. 

In the next section we elaborate on the models we deal with in this paper. Then embark on a sequence of reductions, introduce the Recoverable Reliable Broadcast,
present Reliable Broadcast for the correlated faults model, review related work, and 
close with concluding remarks. But first, we summarize our contributions.

\subsection{Contributions}
\vspace{-0.5em}
\subsubsection{Results}

\beginsmall{enumerate}
\item The $t$-resilient Asynchronous Byzantine model is equivalent to the $t$-resilient Asynchronous Benign model, where at most $t$ of the user's inputs have been changed (i.e. the asynchronous Cuckoo model).
\item The $t$-resilient Asynchronous Byzantine model is equivalent to the a \emph{synchronous}  $t$-mobile model, where at most $t$ of the user's inputs have been changed (i.e. the synchronous Cuckoo model).
\item Extend Correlated Faults to the Asynchronous Byzantine model.
\item Design a Recoverable Reliable Broadcast algorithm where a processor under attack can recover once some attacks stop (signaled by the number of outputs).
\endsmall{enumerate}
\subsubsection{Concepts}
\beginsmall{enumerate}
\item Byzantine models as an attack on the communication subsystem rather than at attack on the processor's program.
\item A model notion as conditions for progress where the condition might change as a function of the outputs already delivered.
\item Ascribe meaning of output to  a ``byzantine processor,'' and require its progress once the attack stops. 
\endsmall{enumerate}

\hide{
We do not stop at importing the Adversary to the Byzantine model. Since asynchronous benign failure models can be modeled as synchronous models with mobile faults, we push the reduction from Byzantine Asynchronous to Byzantine-Mobile Synchronous, and vice versa. We show the synchronous Byzantine-Mobile Model to be exactly the Asynchronous Byzantine model. Since traditional asynchrony and synchrony, now translates into mobile vs fixed failures, it gives rise to questions about possible combination, e.g. a list of mobile-bad-sets and a list of fixed-bad-sets. This Hybrid model was characterized in \cite{hybrid-arxiv16a}.
}


\section{Tasks, Models and Problem Statement}

In this section we first repeat the notion of task, define the models we deal with, 
and state the problem this paper solves.

\subsection{Tasks}
A task is the elementary problem a distributed system implements.
In distributed computing a task play the role a function plays in centralized computing.
The task is a mathematical triple $(I,O,\Gamma)$, where $I$ consists of set of tuples of sets of processors with their inputs, $O$ consists of set tuples of sets of processors with outputs, and $\Gamma$ is a set of
binary relations assigning a set of output tuples from $O$, to any  given input tuple from $I$. This relation is subject to the constraint the the set of processors in the input tuple, and the set of processors in the output tuples match.
Hence, notions like ``correct'' and ``faulty'' that might come from a model, should not be referred to in a definition of a task.

The task prescribes to each participating set of processors with their inputs, what are the output combinations that they are allowed.

The notion of which processors participate in a protocol and which are not is model dependent.
In general, if a set is participating it should be a set of processors that is closed under
the relation ``affected by.''  A participating processor should be oblivious to a non-participating one.

For instance, in the task of leader election on two processors $p_0$ and $p_1$, the input sets
are\\
 $\{ \{(p_0,p_0)\}, \{(p_1,p_1)\}, \{(p_0,p_0),(p_1,p_1)\}\}$,\\
the output sets are\\
 $\{ \{(p_0,p_0)\}, \{(p_1,p_1)\}, \{(p_0,p_0),(p_1,p_0)\}, \{(p_0,p_1),(p_1,p_1)\}\}$,\\
 and $\Gamma$ is\\
  $\{ ((p_0,p_0),(p_0,p_0)), (p_1,p_1),(p_1,p_1)), (((p_0,p_0),(p_1,p_1)),
((p_0,p_0),(p_1,p_0))), (((p_0,p_0),(p_1,p_1)),
((p_0,p_1),(p_1,p_1)))\}$.

\hide{
In words, it prescribes for each set of participating processors with their corresponding inputs,  what combinations of output are acceptable. The notion of participation is model depended. It is easy to define it for 
Shared-Memory, e.g a processor that wrote at least once.  Yet, in models in which a ``write'' is implementable by a whole procedure, then the participation notion has to be precisely defined. On the intuitive level participation has an inductive definition: If a set is participating, then it better be closed under the notion of ``affected by.'' A participating processor should not be affected by a non-participating processor.
}
\subsection{Models}
A model is a setting in which processors learn about each other's input, to produce each an output.
The problem processors face is that they are required to output, under a limited knowledge of other processors input, and correspondingly, limited knowledge what other inputs other processors know.
Hence to learn enough information to output, they are involved in negotiations of telling each other what input they know and updating their knowledge about others. Of course, they cannot negotiate forever and need to output at some point, hence the problem.

Depending on the assumptions about the means of the negotiations, some models are more ``tightly coupled'' than others allowing to solve a larger set of tasks.
In this paper we do not consider systems that cannot solve $\epsilon$-agreement \cite{Rapprox}. In this task we are given a finite interval of $R$ of size at least $1/{\epsilon}$. Processors' inputs are arbitrary integers in the interval and all have to output integers within the convex hull of the inputs such that pairwise they output either the same integer or distinct integers which are adjacent.
We do not consider models which cannot solve $\epsilon$-agreement, 
as we expect any reasonable system to have enough coordination power to solve this task.

Processors negotiate and output. After outputting they linger around to help other processors communicate (essentially by serving in effect as repeaters). 

We take the view of a model  following \cite{AL91} as a contract between a specifier and a programmer
(cf. Rely-Guarantee in \cite{AL91}). Progress, in the form of additional processors outputting, is required in runs of the model where the specifier guarantees hold. 
We insert a new element to the contract not used in the past. The specification may be parametrized by processors' outputs. E.g. in the $t$-resilient case the specifier promises that all processors but $t$ will participate and will  infinitely often take steps in the negotiations (be alive). Under
these conditions, the programmer is to deliver at least all but $t$ outputs. Thereafter the specifier requires an output if at least all but $t$ processors are alive, and at least among the live processors, there is a processor that hasn't outputted yet.

\subsubsection{The Models in this Paper}
We assume a set of $n$ processors 
$\Pi=\{p_1, p_2, ..., p_n\}$. 
The asynchronous settings is the classical asynchronous processors communicating over
a point to point complete network. But, our definition
of Asynchrony and Byzantine is not Classical. While traditionally faulty behavior has been
ascribed to processors, we ascribe it to the communication subsystem. We assume an adversary that can manipulate outgoing messages. 
What the Byzantine and the Benign here have in common is that each has a set of bad-set of processors $\B$, where for all 
$B \in \B$, $B \subset \Pi$ and if $B' \subset B$ then $B' \in B$, i.e the set of bad-sets is subset closed.
A good-set is any set $\Pi \setminus B$.
We denote the set of good-sets $\G$.
In the Benign Fault model
the Adversary can remove messages sent by a set of processors 
from $B$. In the Byzantine case it can 
not only remove messages from $B$ but also alter their content before delivery (but, it cannot generate messages 
on it own).
Traditionally, $t$-resilience stands for a bad-set collection $\B$ that contains any set of cardinality less than $t+1$. 
\hide{
Since we are accustomed to think $t$-resilience rather than Adversary we will stick throughout the paper with the special adversary of $t$-resilience. To see that all we do generalizes to Adversary with a collection of set $\B$, at the end of the paper we show how Reliable Broadcast generalizes to a collection $\B$ that satisfies certain simple condition. It then becomes obvious that all we did does generalize to Adversary.
}

\begin{enumerate}
\item {\bf Benign Asynchronous System:} The adversary can choose a set of processors 
of cardinality of at most $t$ and remove some or all of their messages. 
Otherwise it can delay any message sent by anybody 
for any finite time. This allows the programmer to use the construct ``wait until received messages from a (good-)set of cardinality $n-t$,'' when  all processor are programmed to send messages to all.

Under this conditions, the programmer is obliged to deliver at least $n-t$ outputs. From there on the delivery of outputs hinges on the additional condition that the messages sent by some processor which hasn't 
outputted will not be delayed infinitely long. (Traditionally, some researchers allowed processors which have output to halt. This is a mistake and cannot be sustained in the Byzantine Asynchronous Model. Therefore, in general, we take the view that processors never halt.)

\item {\bf Byzantine Asynchronous system:} In this system in addition to the above, the adversary has the power to change the content of messages sent from a fixed set 
processors of cardinality at most $t$. Notice that we do not allow the adversary to inject messages into the system by itself. (This last condition is restrictive with respect to the view that Byzantine means that a processor was taken over, yet we claim that
nevertheless all
legacy results hold.)

\noindent
{\bf Remark} The traditional threshold Byzantine Asynchronous model motivated a variation of the $\epsilon$-agreement task
in which the at most two integers output have to be in the convex-hull of every $n-t$ inputs \cite{Rapprox}.
It is easy to see that the new variation we may pose in which we require the output integers to be
in the convex-hull of every good set, is also solvable using the same techniques as in \cite{Rapprox}.
And, in the multi-dimensional version of the problem we can replace the $t$ ``bad'' processors by an
adversary that captures $t$ (i.e. $t,n/4$ replace by good sets each 4 of which have nonempty intersection). This
allow to calculate the resulting set consensus power via reduction to the benign adversary.

The \emph{Cuckoo Asynchronous system:} 
This system is a Benign Asynchronous system in which the adversary not only remove messages from a set of cardinality at most $t$ but it also might,
unbeknowethed to the affected processors, change their inputs.  

The \emph{problem statement} of this paper is to show the equivalence between the Byzantine Asynchronous system and the Cuckoo Asynchronous system.

Our notion of equivalence between two models is that they both solve the same set of tasks. In the Byzantine/Cuckoo model, it is up to the user to define this notion of solvability since the adversary might produce an input combination not expected by the user. We claim equivalence whichever way the user will resolve such a notion, as our mechanism of showing the equivalence is by showing that any  protocol in one model can be emulated by a protocol in the other.

\noindent
{\bf Remark:} Once we are in the Byzantine case past $n-t$ outputs we tread in an unchartered territory. Never before the concept of an output from a ``byzantine processor'' 
was considered. What we try to capture is a correct processor, that snoops at others communication while what it sends may be under attack. And then, whether it can make itself whole once the attack ceases (albeit with possibly changed input). When does an attack cease? What the specifier promises the programmer, in this case, is that once $n-t$ outputs are delivered, one of the processors that have not output yet will cease to be under attack. From there on its messages would not be altered any more and all its messages sent thereafter will eventually be delivered. We use the term ``become correct'' for such a processor, and ``correct'' for a processor whose messages are never attacked by the adversary. Observe that after a processor becomes correct, the adversary can still choose to deliver old and manipulated messages the processor had sent prior to recovering (becoming correct).

3. {\bf Synchronous Mobile Adversary:}  We show that the two asynchronous systems above are equivalent to synchronous systems. The Benign Asynchronous is equivalent to a synchronous system 
in which the choice of at most $t$ processors from which the adversary may remove messages is a choice \emph{per round}, hence the name ``mobile'' adversary. We also show that the Byzantine Asynchronous system
is equivalent to synchronous system
as above, only that now the adversary can also tamper with messages
rather than just remove them. The point of the paper is that this additional power of tampering can be confined to appear just as changing some $t$ inputs at most, and from there on the adversary can just remove messages, i.e resulting in the Synchronous Mobile Cuckoo Model.  We denote these systems as \BIMObb where in the benign asynchronous case the input set to the synchronous equivalent system is the empty set.

\noindent
{\bf Remark:} The gap between the Asynchronous system we presented and the Synchronous Mobile is  large.
In the latter, in a round, the missed messages are all confined to messages from a single set of at the most cardinality $t$. In the former, the asynchronous case, each processor might miss at most $t$ messages too, but the missed messages are not confined to come from the same set of cardinality at most $t$.
The direct synchronous analog to the asynchronous system is to allow the adversary to attack
the reception process rather than the transmission process: For every process $p_i$ in a round, the adversary can choose at most $t$ incoming messages to be removed. This leads us to the synchronous system
denoted by \BISYNt. The \BISYNt will be mediating between the Asynchronous and the Synchronous breaking the final reduction into two intermediate reductions.
\end{enumerate}

\subsection{The Correlated-Faults Adversary Model}

In \cite{HirtM97,AFM99} it was suggested to replace the condition for synchronous Byzantine agreement from $n \geq 3t+1$ to a list of faulty set $\B$ under the condition that no union of three sets from $\B$ is $\Pi$. By showing that this condition is enough to implement Reliable Broadcast all our results cast in the resilience model carry verbatim to the correlated case with all the implications and results of correlated case in the Benign model (see \sectionref{sec:RB}). 
Obviously if the condition is violated and the union of some three sets is $\Pi$, then the Adversary can cause network partition with no possibility of a meaningful computation.

\section{Simulating a protocol running in \BISYNt system}

In this part we  emulate
 any protocol running in a \BISYNt system by a protocol running in a \ttbasyn.  
Thus, we prove that the adversary in an asynchronous system with the $t$  processors it controls and with the ability of delaying messages to all processors does not have any more power than the adversary  in the synchronously running  \BISYNt system. 


The idea behind the simulation is to completely simulate the original protocol running in a \BISYNt system  by a protocol running at each processor  in a \ttbasyn, such that  every  processor that becomes correct 
in the \ttbasyn would get the same output as it would have gotten in the \BISYNt system. We assume that initially there are $n-t$ correct processors, and new processors are relieved from the adversary's control and become correct during the run of the protocol.
 
In the emulation appearing in this section we make use of two elements.
The first element employed is to make sure that every processor commits to the message it sends in each round in a way that if any processor accepts a message $m$, everyone will eventually accept the same message $m$, and before accepting $m$ it accepts  all $m$'s causally ordered prior messages. 
Thus, the first element, \cosend (\algref{figure:rbg}), is a {\em Causally Ordered Recoverable Reliable Broadcast primitive}, which uses a generalization of the classical reliable broadcast as a building block, combined with the causality of message delivery (generalizing a combination of~\cite{Bracha:1984,DistComp}). 
As we detail below, to further limit the ability of the adversary to confuse  processors we instruct each processor, after sending its initial input, to send only the list of IDs of the processors it claims to accept and process messages from in the previous round, without any additional content (we differentiate between receiving a message, agreeing to accept it as a legitimate message to be processed, and processing it).
The second element is to locally simulate the state of the original protocol at every other processor, 
according to the list of processors' IDs being accepted, to determine what the original (simulated) protocol instructs each processor to do, what values should be sent and which should be received.

Let $P$ be a deterministic message passing protocol that is executed  in a \BISYNt system.  We will show a simulation of $P$ in a \ttbasyn.  In the simulation, in the first round each processor, $p$, uses $\cosend(1,p)$ to broadcast  its own input value, $\I$.  
In each subsequent round, $r$, each processor, $p$, uses $\cosend(r,p)$  to broadcast  to everyone a set, $\pi_{r-1}$, of the IDs of the processors from which it accepted and processed messages in the previous round, round $r-1$. 

We start with an overview of the simulation protocol (\algref{figure:state-sim}).
Each processor, $p$, maintains locally $n=|\Pi|$  state machines, $SM_i$, $1\le i\le n$, to reflect the state machines of the original (simulated) protocol at each other processor.
When $p$ accepts via \cosend a message $\tri{1,\I}$ from a processor $p_i$, it initiates state machine  $SM_i$ with the input value $\I$. 
If $p$ does not accept such a message from a processor, say $q$, it does not start the state machine $SM_q$.  
The properties of  \cosend ensure that if any  processor will ever accept a message  $\tri{1,\I}$ from a processor $q$, then eventually processor $p$ will also accept it too.

\begin{algorithm}[!t]
\footnotesize
\SetNlSty{textbf}{}{.}
\setcounter{AlgoLine}{0}
\  \hfill\textit{/* executed at processor $p$ */}\\
\lnl{line:def}  \     {\bf set}  $\forall k$ $accept[k]:=\emptyset$;  \hfill\textit{/*  the sets of  senders whose messages were processes at the given rounds */}\\
\lnl{line:def2}  \     {\bf set}   $\M:=\emptyset;$ $\bar\M:=\emptyset$; \hfill\textit{/* the set of accepted  messages that were not processed yet , and the set of processed messages */}\\
\lnl{line:def3}  \     {\bf set}   $O:=\emptyset;$ ;\hspace{2mm} \hfill\textit{/* the set of processors that their SM reached an output at the  processor */}\\

\ \\
\lnl{line:init-r} \    $r:=1;$\hfill\textit{/* the round number */}\\
\lnl{line:step1} \    {\bf invoke }  $\cosend(r,p)$ to broadcast $\I$;\mbox{\ }\hfill\textit{/* broadcast the input value, a processor sends  also to itself */}\\

\ \\
\nl \    {\bf do }  {\bf until} $SM_p$ halts:\hfill\textit{/* participate in all $\cosend(\ell,*),$ $\ell \le r$, protocols */}\\

\lnl{line:collect}\    \tb {\bf wait until } $|accept[r]|\ge n-t$ and $p\in accept[r]$; \\
\lnl{line:core}\ \mbox{\ }\hfill\textit{/* ``empty line" - for later use */}\\
\lnl{line:inc}\    \tb $r:=r+1$;\\
\lnl{line:send} \    \tb {\bf invoke } $\cosend(r,p)$ to broadcast $accept[r-1]$;\mbox{\ }\hfill\textit{/* broadcast the accepted and processed set in round $r-1$ */}\\

\lnl{line:end} \    {\bf end.} \\
\ \\
  {\underline {\bf In the Background:}} \\
\lnl{line:r1} Execute for each $\tri{r',p_i,\pi}\in\M$: \hfill\textit{/*  accepted message from $p_i$ for round $r'$   */} \\
\lnl{line:init-rec}    \tb {\bf if} $r'=1$ {\bf then start } $SM_i$ with input $\pi$;\hfill\textit{/* start a SM with the initial input */}\\
\lnl{line:init-rec2}     \tb {\bf if} $r'>1$ {\bf then } \\
\lnl{line:rec-m}  \tb\tb {\bf let} $M:=\{m_j\mid p_j\in\pi \mbox{ and } SM_j[r'-1] \mbox{ sends } m_j \mbox{ to } p_i\}$; \hfill\textit{/* the values $p_i$ should have received */}\\
\lnl{line:apply-m}  \tb\tb$SM_i[r']:=\F(M,SM_i[r'-1],r')$; \hfill\textit{/* apply  protocol $\F$ to determine the next state of $SM_i$ */}\\
\lnl{line:halted}  \tb\tb {\bf if } $SM_i$ outputs {\bf then} $O:=O\cup \{p_i\}$; \hfill\textit{/* Denote the outputting processor */}\\
\lnl{line:r3}    \tb $\M:=\M\setminus \tri{r',p_i,\pi}$;\\
\lnl{line:r4}     \tb $\bar\M:=\bar\M\cup \tri{r',p_i}$;\\
\lnl{line:add}       \tb   $accept[r']:=accept[r']\cup\{p_i\}$. \\
%
\caption[caption]{Simulating a \BISYNt system deterministic protocol   \\\hspace{\textwidth}\mbox{\ \hspace{0.83in}} in an asynchronous  system with $n>3t$.}\label{figure:state-sim}
\end{algorithm}

When processor $p$ accepts via \cosend a message $\tri{2,\pi_{1}}$ from a processor $p_i$, it uses the initial input values of all $q_j\in\pi_1$  
as the set of  input values to $SM_i$ to determine what values the simulated protocol $P$ instructs processor $p_i$ to send to every other processor in round $2$ of protocol $P$. 
Since \cosend implements  Causally Ordered Reliable Broadcast, before processor $p$ processes $\tri{2,\pi_{1}}$  from a processor $p_i$, it already accepted and processed all $\tri{1,\I}$ messages from all processors in $\pi_{1}$, and therefore knows the values $p_i$ should have received from the members of $\pi_i$.
Observe that if the adversary causes the message from processor $p_i$ to claim to accept an input from a processor, say $q$, that did not send it an input, then the 
$\tri{2,\pi_{1}}$ message received from $p_i$ will be put aside and the \cosend invoked by $p_i$ will not be completed at any other  processor until (and if) the input from $q$  will be processed.

Now recursively, when processor $p$ accepts via \cosend a message $\tri{r,\pi_{r-1}}$ from a processor $p_i$, it uses the  values every $q_j\in\pi_{r-1}$ should have sent in round $r-1$ of protocol $P$ to $p_i$ according to $q_j$'s state machine $SM_j$ at round $r-1$,  as values received by $SM_i$ in the previous round (round $r-1$)   to determine what  values the simulated protocol $P$ instructs processor $p_i$ to send to every other processor in round $r$. 

The simulation protocol, presented in \algref{figure:state-sim}, maintains  four data structures, $\M, \bar\M$, $accept$ and $O$.  
The set $\M$ contains the messages that were accepted via \cosend and were not processed yet. There is at most one
such message per sender (because all messages of the same sender are casually related).  
Each entry in $\M$ contains a round number, say $r$, a processor ID, say $p_i$, and the set of processors' IDs, from which processor $p_i$ claims to have accepted and processed messages in round $r-1$.

The set $\bar\M$ contains the list of processors whose messages were already processed by processor $p$, and their $SM$s are updated accordingly. 
Each entry in $\bar\M$ contains a round number, say $r$, a processor ID, say $p_i$, indicating that round $r$ message from $p_i$ was accepted and processed. By the \cosend properties, every processor that processes a round $r$ message from $p_i$ processes the identical message.

The \cosend\ protocol, \algref{figure:rbg},  is a  communication procedure that exchanges messages and holds received messages until they are ready to be accepted and  processed. The \cosend properties (as we later explain) imply that when an entry is added to $\M$, all casually prior entries were already accepted and processed (thus, are already in $\bar\M$ ), and as such are reflected in the respective state machines (as we explain later). 
Therefore, each message in $\M$ can be processed independently, since there are no causal dependencies among them.
Processing a message is just applying it to the state machine of the sending processor, using the current state of the state machines of all
the processors it claimed to accept their messages in the previous round.
Once a message is processed it is removed from $\M$ and added to $\bar\M.$
Observe that the simulation may indicate that at a certain round some processor is not sending a value to some other processor, then in such a case no such value is produced as an input to the relevant state machine.

The third data structure ($accept$) is the set of processors whose messages were accepted and processed in the given round. 
$accept[r]$ is the  current set of all the processors whose round $r$ messages were accepted and processed by $p$  via   \cosend.   
Once 
$|accept[r]|\ge n-t$, and $p\in accept[r]$, processor $p$ uses $\cosend(r+1,p)$ to broadcast the set $accept[r]$. After broadcasting this message processor $p$ continues to accept all rounds' messages via  \cosend and continues to apply them to the various state machines.
Each processor continues this process, outputs its output, and sends messages until its state machine halts.\footnote{The processor continues to participate in the \cosend protocols of other processors even after it halts. One can add a halting task that enables a processor to halt once enough other processors reach the right stage.}  

The fourth data structure is the list of processors whose state machines produced  outputs at the processor. Whenever the processor learns that some state machine produced an output it adds it to the set $O$.

\begin{algorithm}[!ht]
\footnotesize
\SetNlSty{textbf}{}{.}
\setcounter{AlgoLine}{0}
 $\M$, $\bar \M$ and $O$ are globally maintained sets \hfill\textit{/* executed by processor $p$ with  sender  $s$ in round $r$, invoked once per round */}\\


\ \\
{\bf Sender's Protocol:}\mbox{\ }\hfill\textit{/* $s$ is the sender and $v_s$ the value it broadcasts */}\\
\lnl{line:cog-step0}   \tb The sender $s $  invokes RecRB$(s)$ with $v_s$ to send to all.\mbox{\ \ }\hfill\textit{/* RecRB is a recoverable Reliable Broadcast (see \sectionref{sec:RRB}) */}\\

\ \\
{\bf Any  Processor's Protocol:}\\
\lnl{line:cog-test}  \tb  {\bf Upon receiving} an RecRB$(s)$ protocol message with value $v$: \\
\lnl{line:cog-ex1a}   \due {\bf if } ($r>1$ ) {\bf wait until} $\;s\in v$ {\bf and} $\forall q\in v, \tri{r-1,q}\in \bar\M$; \hfill\textit{/*   wait for the causally prior messages   */} \\
\lnl{line:cog-exec2}    \due join RecRB$(s)$ as participant; \hfill\textit{/* joined at most once per sender per round */}\\
\ \\
\lnl{line:cog-exec2a}  \tb  {\bf case}  accepted RecRB$(s)$ with value $v'$: \\
\lnl{line:cog-ex2a}   \due {\bf if } ($r>1$ ) {\bf wait until} $\;\forall q\in v', \tri{r-1,q}\in \bar\M$; \hfill\textit{/*   wait for the causally prior messages   */} \\

\lnl{line:cog-acpt} \due{\bf Accept:}  $\M:=\M\cup\tri{r,s,v'}$.\hfill\textit{/*   accept message $v'$ as the message sent by processor $s$  in round $r$  */} \\
\ \\
 {\bf Case} A new processor is added to $O$, \ie outputs:\mbox{\ }\hfill\textit{/* Try to recover */}\\
\lnl{line:cog-retest}\tb {\bf if} $|O|\ge n-t$ {\bf then} sender $s$ repeats the last send of \lref{line:cog-step0}.\mbox{\ }\hfill\textit{/* repeat the last sending */}\\
\caption[caption]{ \cosend\!$(r,s)$: A casually ordered reliable broadcast \\\hspace{\textwidth}\mbox{\ \hspace{0.841in}} with asynchronous  system with $n>3t$.}\label{figure:rbg}
\end{algorithm}

Observe that the  emulation, presented in \algref{figure:state-sim}, produces per each processor an agreed upon sequence of sets of values  $M_0, ... M_{r-1}$ processed by its SM in the related rounds, thus, simulating the exact behavior of protocol $P$. This implies that the above emulation is a protocol to simulate in a \ttbasyn a deterministic message passing protocol, $P$, in a \BISYNt system.

The enabling properties of the simulation reside in the details of \cosend, which we now describe. 
The \cosend\ protocol, \algref{figure:rbg}, is invoked per processor per sending round and consists of 3 conceptual parts.  
In the first part the sender of the current instance of the protocol    (in \lref{line:cog-step0})
invokes RecRB, a recoverable Reliable Broadcast (see \sectionref{sec:RRB}), to send its  value to everyone.

RecRB ensures that:
\def\rrbdef{
\beginsmall{enumerate}
\item[RRB1] If the sender $s$ is correct when it invokes RecRB with value $v$,  then every processor will eventually accept RecRB$(s)$ with the value $v$. 
\item[RRB2] If any  processor accepts RecRB$(s)$ with a value $v'$, then every processor will accept it with the same value.
\item[RRB3] If sender $s$ did not invoke RecRB$(s)$ no processor will accept RecRB$(s)$.
\item[RRB4] If the sender $s$ becomes correct after invoking the protocol, then every processor will eventually accept RecRB$(s)$.
\endsmall{enumerate}
}
\rrbdef

The second part is executed by every processor. In the first round, every processor joins the  RecRB invoked by the sender.  In later rounds processors validate the value, since the value is a set of processors whose messages were accepted in the previous round. We need to ensure that if the RecRB will be accepted, the resulting value will be consistent with the local views of all processors.  To ensure that, if a processor learns that the RecRB was invoked by $s$ it waits until the value sent is contained in its own view of the previous round (\lref{line:cog-ex1a}). It joins the invoked  RecRB  only when this holds.  It joins a single invocation per sender per round. The local testing ensures that the accepted value will eventually enable every processor to apply the accepted value to the local state machine of $SM_s$.  
Since the accepted value may differ from the value first seen by the processor, the participating processor waits again  (\lref{line:cog-acpt}) until its local view becomes consistent with the accepted value.  This is guaranteed to happen eventually, since the accepted value tested by any processor will eventually hold at any other processor.

If $s$ is not correct when it invokes the protocol, the protocol may get blocked. For example, if the local test in (\lref{line:cog-ex1a}) fails at all processors.  If $s$ is recovered (becomes correct) while running this protocol, when it resends its value in  (\lref{line:cog-step0}) the protocol recovers. The third part intends to catch exactly that.   Every time a new processor recovers while running  \cosend,  its  \cosend will complete successfully. The adversary may cause each processor it controls to resend values at most $2t$ times.

Given the above discussion, it is clear that if the sender is correct, all  processors will eventually  complete the protocol and will accept its value.
Moreover, if any  processor accepts a message, every  processor will end up eventually accepting the same message, after accepting and processing all causally prior messages to that message.
Thus, we outlined the proof of the following claim.


\begin{lemma}\label{lem:co}
If $n>3t$ , then
\algref{figure:rbg} implements 
a Causally Ordered Recoverable Reliable Broadcast transport layer in which if a sender $s$ uses  \cosend to send its messages, each  processor accepts messages that satisfy:
\beginsmall{enumerate}
\item[\textup{ CO1:}]
If a correct sender, s, sends a consecutive sequence of messages, then every processor accepts the sequence in the same order that it was sent.

\item[\textup{ CO2:}] For $r>1$, if a  processor, $p$, accepts  a $v$  via $\cosend(r,s)$, then $s\in v$ and $p$ already accepted $v_j$  via $\cosend(r-1,p_j)$,  for every $p_j\in v$.

\item[\textup{ CO3:}] If a  processor, $p$, accepts  a $v$  via $\cosend(r,s)$, 
then every  processor will accept $v$  via $\cosend(r,s)$.


\endsmall{enumerate}
\end{lemma}

The above discussions, \lemmaref{lem:co}, and given that a processor moves to the next round, once it accepts and processes current round messages from a set of $n-t$ processors, imply the following result.

\begin{lemma}\label{lem:BISYNt}
Given  a deterministic protocol 
$P$ that is viewed as a function $\F(M_{r-1},S_p,r)$, for $r\ge 1$, 
in a \BISYNt system,
the emulation protocol presented in \algref{figure:state-sim}  
simulates it at every processor in a \ttbasyn, provided that $n>3t$ .
\end{lemma}

\proof
To prove the claim what we need to show is a mapping of processors and inputs from a \ttbasyn to the corresponding  processors and inputs in a \BISYNt system; and show that every run in the \ttbasyn corresponds to a possible run in the \BISYNt system and that  correct processors in the \ttbasyn  obtain the same output the corresponding processors obtain in the \BISYNt system.


Let  $B$  be the set of Byzantine processors the adversary initially controls.  
Let $G=\Pi\setminus B$  be the appropriate set of correct processors.
We map the processors' IDs the same in both systems, and we will consider those processors controlled by the adversary in the \ttbasyn as the Cuckoo set,  those receiving a bad input in the \BISYNt system.

Assign to a processor in the \BISYNt model the input value a correct processor obtained  from the respective  \cosend in \lref{line:step1} of \algref{figure:state-sim}, when running it in the \ttbasyn (processors that never take an action in the \ttbasyn are part of $B$ and will be considered as processors in the \BISYNt model whose messages never arrive to any other processor).
Thus, at the beginning of the first phase, once the inputs are introduced, the state of the corresponding  processors are the same. 

The set of IDs each processor receives (and accepts) in the first round is the set it sends when executing \lref{line:send}.
The state of each processor $p$ at the end of round 1 is the state it obtains when processing its round 1 message when executing \lref{line:apply-m}.
Observe that every correct processor eventually complete the first round and send its initial input.
Thus, at the end of round 1, the local state of  processor $p$ is the same state its state in \BISYNt model at the end of round 1.
Observe that, by \lemmaref{lem:co}, any other processor that computes locally the state of processor $p$, when it processes processor's $p$ round 1 message when executing \lref{line:apply-m} will obtain the same state at its local state machine of processor $p$.

Now, by induction, we can show that the set of processors each processor $p$ accepts in round $r$ is what it sends in \lref{line:send}, contains a set of $n-t$ processors.  Its state  after processing the round $r$ messages (at the beginning of round $r+1$) is what every processor obtains when processing its broadcasted message when executing  \lref{line:apply-m}.

When the protocol in the \BISYNt model instructs a  processor $p$  to produce an output the SM at every processor will  produce in the \ttbasyn the identical output. Observe that in our model all  processors will run all state machines of all initially  correct processors, and thus will produce the outputs. Once a SM produces an output at any processor, it eventually produces the same output at every other processor. Once a  processor learns about an output  (executing \lref{line:halted} of \algref{figure:state-sim})  it updated its list of outputs. 
Thus, once all initially correct processors output, every processor learns about that, and will try to recover.  One of them will recover and will reach an output state. The process repeats itself, and eventually all processors output on all state machines.
\endproof

\section{Recoverable Reliable Broadcast}\label{sec:RRB}

Traditional Reliable Broadcast ensures properties [RRB1] - [RRB3], below.
The adversary may cause a traditional Reliable Broadcast protocol to block when it is invoked by a processors it controls, by sending conflicting values to different processors.  Our aim is to ensure that once a processor recovers, the protocol it previously invoked while being controlled by the adversary will recover and complete. To obtain that we extend the properties of Reliable Broadcast to be:

\rrbdef
%

The Recoverable Reliable Broadcast protocol presented in this section (\algref{figure:rrb}) obtains these properties.  The protocol uses as a building block the classical Reliable Broadcast protocol  (\algref{figure:rb}, see~\sectionref{sec:RB}), that was developed by Bracha~\cite{Bracha:1984}, and was used by many researchers ever since.

\begin{algorithm}[!ht]
\footnotesize
\SetNlSty{textbf}{}{.}
\setcounter{AlgoLine}{0}
$O$ is a globally maintained sets \hfill\textit{/* executed by processor $p$ with  sender  $s$, invoked once per round per sender */}\\

 \ \\
 \lnl{line:rrb-def}  \     {\bf set}  $H[k]=\emptyset$ for any $k \ge 0$\hspace{2mm} \hfill\textit{/*  $H[k]$ includes all accepted RB of round $k$ */}\\

\ \\
{\bf Sender's Protocol:}\\
\nl \tb $k:=0$; $v:=v_s$;\\

\lnl{line:rrb-loop1} \tb   {\bf repeat }\\ 
\nl \due $k:=k+1$;\\
\lnl{line:rrb-sendall} \due {\bf send} $\ang{v,k,H[k-1]}$  to all;\mbox{\ }\hfill\textit{/* $s$ is the sender and $v$ the value it attempts to broadcast */}\\
\lnl{line:rrb-loop1} \due   {\bf wait until}  $| H[k]|\ge n-t$ \mbox{\ }\hfill\textit{/* $n-t$ different RB were accepted */}\\
\lnl{line:rrb-loop2} \tre  {\bf if}  $\exists $ unique $\bar v, $ s.t. $|\{q\mid (q,\bar v)\in H[k]\}|\ge t+1$ {\bf then} $v:=\bar v$ {\bf else} $v:=v_s$;\mbox{\ }\hfill\textit{/* either $\bar v$ or the input value */}\\
\lnl{line:rrb-loop3}\tb {\bf until}  $\exists k',v', $ s.t. $|\{q \mid (q,v')\in H[k']\}|\ge n-t$ \mbox{\ }\hfill\textit{/* $n-t$ different RB were accepted with a value $v$ */}\\
\ \\
{\bf Any Processor's Protocol:}\hfill\textit{/* the sender also run this protocol */}\\
\lnl{line:rrb-exec1}    \tb {\bf case}  {\bf received} $\ang{v,k,H'[k-1]}$ from $s$:  \\
\lnl{line:rrb-wait-to-join} \due {\bf wait until} $H'[k-1]\subset H[k-1]$ \hfill\textit{/* the sender also run this protocol */}\\
\lnl{line:rrb-wnew-rb} \tre {\bf if} (did not invoked RB for $k$) and ($\not\exists v'\not=v$   s.t. $|\{q\mid (q,v')\in H'[k-1]\}|\ge t+1$) \\
\lnl{line:rrb-wnew-rb2} \tre\tb {\bf then} invoke RB$(k,p)$ with value $v$ ;\hfill\textit{/* executed at most once per $k$  */}\\
\ \\
\lnl{line:rrb-exec2}    \tb {\bf case}  {\bf accepted}  $\bar v$ via RB$(k,q)$: \\
\lnl{line:rrb-add-rb}  \due add $(q,\bar v)$ to $H[k]$; \\
\ \\
\lnl{line:rrb-exec3}   \tb {\bf case}  
 $\exists k',v', $ s.t. $|\{q \mid (q,v')\in H[k']\}|\ge n-t$: \mbox{\ }\hfill\textit{/* $n-t$ different RB were accepted with a value $v$ */}\\
\lnl{line:rrb-acpt} \due{\bf accept}  RecRB$(s)$ with value $v'$.\hfill\textit{/*   accept message $v'$ as the message sent by processor $s$  via RecRB  */} \\
\ \\
 {\bf Case} A new processor is added to $O$, \ie outputs:\mbox{\ }\hfill\textit{/* Try to recover */}\\
\lnl{line:rrb-retest}\tb {\bf if} $|O|\ge n-t$ {\bf then} sender $s$ repeats the last send of \lref{line:rrb-sendall}.\mbox{\ }\hfill\textit{/* repeat the last sending */}\\
\caption[caption]{ RecRB$(s)$: Recoverable Reliable Broadcast \\\hspace{\textwidth}\mbox{\ \hspace{0.841in}} with asynchronous  system that satisfies  the Byzantine Fault predicate}\label{figure:rrb}
\end{algorithm}

The challenge is in ensuring  that despite the degree of freedoms the adversary has in trying to confuse processors by changing the content of messages it can't block  recovered processors from making progress. The high level idea of the protocol is that the sender will repeatedly push forward a value via a sequence of sending attempts, until any of the attempts collects support for the same value from at least $n-t$ processors. Each attempt  brings every processor to invoke a Reliable Broadcast to send the value it received from the sender for that attempt. 
Each participant maintains a data structure in which it collects the accepted Reliable Broadcasts and the values it accepted for each attempt.

The protocol is composed of three parts.  In the first part the sender sends in each iteration a value and the set ($H[k-1]$) of accepted Reliable broadcasts of the previous iteration (\lref{line:rrb-sendall}). In the first iteration this set is the empty set. The processor waits (\lref{line:rrb-loop1}) until it accepts at least $n-t$ Reliable Broadcasts.  Once this happens it determines (\lref{line:rrb-loop2}) the value it should send in the next attempt. The value chosen is either a unique value appearing at least $t+1$ times or the original value the sender started the protocol with.  Notice that if the sender was correct when it invoked the protocol the value is always the  value it started with. The sender repeats this process until for any attempt (\lref{line:rrb-loop3}) there are $n-t$ identical values.  In such a case it can stop, and it knows that every processor will end up seeing that value in at least one iteration, as we argue below.

The second part is  executed by every participant, including the sender.  When a processor receives from the sender (\lref{line:rrb-exec1}) a value and a set of the accepted Reliable Broadcasts at the previous iteration, as claimed by the sender, it carries various validation tests prior to invoking the next Reliable Broadcast. It first waits until the set of Reliable Broadcasts if accepts ($H[k-1]$) contains the set ($H'[k-1$) of the sender.  Next it checks whether the value received from the  sender does not contradict  the set $H'[k-1]$. If both hold and the processor did not already invoke a Reliable Broadcast with index $k$ it invokes one. The processor collects all accepted Reliable Broadcasts in its data-structure $H[k]$, whenever any is accepted, in the context of the current RecRB (\lref{line:rrb-exec2}), until it identifies that its data structure contains at least $n-t$ identical values for some iteration (\lref{line:rrb-exec3}). Once this happens it accepts that value as the sender's value of the invoked RecRB. Because of recovering, a processor may get more than a single message per iteration (at most $t$ such messages in total).  It invokes the corresponding reliable broadcast on at most one of these messages.

Observe that if any processor accepts a value for the current RecRB, that value will be accepted by every processor.  The reason is that the Reliable Broadcast properties imply that what  one data structure contains, every data structure will eventually contain.  Moreover, we argue that there will not be two different values accepted.  Let $\bar k$ be the minimal index for which any processor eventually collects at least $n-t$ accepted Reliable Broadcasts for some value $\bar v$, then the value being sent by the sender in iteration $\bar k +1$ can contain only $\bar v$.  If it differs from that value (due to adversarial manipulation), then it will fail the tests of  \lref{line:rrb-wait-to-join} and \lref{line:rrb-wnew-rb} at any processor, and would never be sent by any processor. 

The third part intends to help the sender to complete the protocol in the case it recovers and its previous message caused the participating processors not to invoke the corresponding Reliable Broadcasts. When it identifies that the number of processors that gave output is more than $n-t$, it knows that there is a chance it will be recovered; and in such a case it repeats the last sending each time it learns about another processor outputting.

Observe that if the sender was correct before the invocation of  RecRB, then eventually every processor will see $n-t$ identical copies of its value in the first iteration, or in any later iteration.  Also notice, that if any processor accepts the RecRB of the sender, everyone will eventually accept it with the same value. Thus, proving the 4 properties of RecRB.

\section{Simulating a protocol in a \BIMObb system}

To finalize the main result of the paper we will now expand the simulation from simulating a 
protocol $P$ that runs is a  \BISYNt  system to a protocol in a   \BIMObb  system.  The extension is to ensure that before a processor adapts the $accept$ set it communicates with others to converge to $accept$ sets such that all sets have at least $n-t$ processors in common.
To achieve that we introduce another element,  we run a couple of  rounds of the equivalent to a full information message exchange to make sure that everyone shares messages from a set of at least $n-t$ processors.  Once this happens the processor takes its next step.

\ccore, is an adaptation of the Get-Core approach mentioned in~\cite{DistComp} (attributed to the second author) and a variation of it that was later presented in~\cite{Abraham2010} as Binding Gather, and using ideas from~\cite{Abraham2005}.

Each processor invokes the \ccore protocol, appearing in \algref{figure:ccore},  with a set of at least $n-t$  different processors IDs. Each processor, $p$,  returns as an output a set of at least $n-t$  different processors' IDs, such that at least $n-t$ of them are shared by the outputs of all  processors. The \ccore  properties are:

\beginsmall{itemize}
\item {\em Validity}: At each  processor, the output set of IDs contains the input set of IDs.
\item {\em Commonality}: There exists a set of $n-t$ IDs  that appears in the output set of every  processor.
\item {\em Termination}: All correct processors eventually output some non-empty set of IDs.
\endsmall{itemize}

A set that is in every output set is called a common core. 
The \ccore primitive is  described in  \algref{figure:ccore}. In the first round, everyone sends its $accept[r]$ set.  
In the background the processor continues to update its $accept[r]$ set with messages it continues to accept.
To complete the first round of \ccore, it waits to receive at least $n-t$ sets that are contained in its current state of the set $accept[r]$.  This will eventually happen due to the \cosend properties, and the fact that there are at least $n-t$ correct processors.  Once this happens it sends again its current set and waits again to received at least $n-t$ second round sets that are contained in its current state of the set.  

Observe that a recovering processor does not need to repeat the last sending as  it did in  previous protocols..

\begin{algorithm}[!ht]
\footnotesize
\SetNlSty{textbf}{}{.}
\setcounter{AlgoLine}{0}
 \ \hfill\textit{/* the input set $accept[r]$ is updated contineously in the background according to the messages accepted  via \cosend and processed in  \algref{figure:state-sim} */}\\
 \ \\
\lnl{line:cc-step1} \  {\bf step 1}  $send(r,1,accept[r])$ to all;\mbox{\ }\hfill\textit{/* send the input set to all, a processor sends also to itself */}\\
\lnl{line:cc-wait1} \  \tb  {\bf wait until} $|\{j\mid received(r,1,\pi_j) \mbox{ from } p_j, \mbox { and } \pi_j\subseteq accept[r] \}|\ge n-t;$\\
\ \\
\lnl{line:cc-step2} \  {\bf step 2}  $send(r,2,accept[r])$ to all;\mbox{\ }\hfill\textit{/* the set $accept[r]$ is being continuously updated in the background */}\\
\lnl{line:cc-wait2} \  \tb  {\bf wait until} $|\{j\mid received(r,2,\pi_j) \mbox{ from } p_j, \mbox { and } \pi_j\subseteq accept[r] \}|\ge n-t;$\\
\ \\
\lnl{line:cc-return}       {\bf return}  $accept[r]$. \\
\caption[caption]{$\ccore(accept[r])$, the Common Core protocol  }\label{figure:ccore}
\end{algorithm}

The correctness proof is similar to the proof of get-core in \cite{DistComp}. The original proof did not include Byzantine nodes, therefore we need to change it a bit.  Define a table $T$ with $n-t$ raws and $n-t$ columns, that refer to a set $G$ of $n-t$  correct  processors.  
The $accept[r]$ value of each correct processor contains at least $n-t$ IDs, therefore it contains at least $n-2t\ge t+1$ IDs of processors in $G$ represented in $T$.  
For $p_i,p_j\in G$, entry $T[i,j]$ in the table is $1$ if $p_j$ is one of the $n-t$ processors that $p_i$ waited for in order to complete the first round of \ccore, and $0$ otherwise.  Observe that if $1$ appears in  entry $T[i,j]$, the $accept_i[r]$ sent by $p_i$ in the second round contains all the $n-t$ IDs appearing in 
the initial $accept_j[r]$ sent by $p_j$  in the first round of the \ccore protocol.

Since all correct processors will eventually invoke \ccore, $T$ will contain  at least $(n-t)(t+1)$ entries with $1$. This implies that there is an ID of a correct processor, say $\bar p$, that appears in at least $t+1$ raws.  
Thus, there are at least $t+1$ correct processors whose second round set includes the $n-t$ IDs that appear in the initial set of $\bar p$.  
Before completing the protocol, each processor waits to get the sets of $n-t$ processors, so it includes the set of at least one of these $t+1$ processors, thus includes the set of $n-t$  IDs appearing in the initial list of $\bar p$.

\begin{lemma}\label{lem:ccore}
For $n>3t$, the protocol presented in \algref{figure:ccore}  
implements the \ccore properties. 
\end{lemma}

To obtain the  final protocol we add the \ccore invocation to the simulation protocol presented in \algref{figure:state-sim}.
We invoke the \ccore protocol on all the  $accept$ sets of a given round after completing \lref{line:collect} and before executing \lref{line:send} of \algref{figure:state-sim}. The output of the \ccore is used in \lref{line:send}  as the set of processors from which we received messages from in that round. \algref{figure:full-state-sim} presents the complete protocol.

\begin{algorithm}[!t]
\footnotesize
\SetNlSty{textbf}{}{.}
\setcounter{AlgoLine}{0}
\lnl{line:f-def}  \     {\bf set}  $\forall k$ $accept[k]:=\emptyset$; \hspace{2mm} \hfill\textit{/*  the sets of accepted senders at the given rounds; executed at processor $p$ */}\\
\lnl{line:f-def2}  \       {\bf set}   $\M:=\emptyset;$ $\bar\M:=\emptyset$;\hspace{2mm} \hfill\textit{/* the set of accepted  messages that were not processed yet , and the set of processed messages */}\\
\lnl{line:f-def3}  \     {\bf set}   $O:=\emptyset;$ ;\hspace{2mm} \hfill\textit{/* the set of processors that their SM reached an output at the processor */}\\
\ \\
\lnl{line:f-init-r} \    $r:=1;$\hfill\textit{/* the round number */}\\
\lnl{line:f-step1} \     {\bf invoke }  $\cosend(r,p)$ to broadcast $\I$;\mbox{\ }\hfill\textit{/* broadcast the input value, a processor sends also to itself */}\\

\ \\
\nl \    {\bf do }  {\bf until} $SM_p$ halts:\hfill\textit{/* participate in all $\cosend(\ell,*),$ $\ell \le r$, protocols */}\\
\lnl{line:f-collect}\     \tb {\bf wait until } $|accept[r]| \ge n-t $ and $p\in accept[r]$; \\
\lnl{line:f-p-core}\  \tb  $accept[r]:= \ccore(accept[r])$ \hfill\textit{/* {\bf the 2 rounds protocol to converge to shared $n-t$} */}\\
\lnl{line:f-inc}\    \tb $r:=r+1$;\\
\lnl{line:f-send} \    \tb {\bf invoke } $\cosend(r,p)$ to broadcast $accept[r-1]$;\mbox{\ }\hfill\textit{/* broadcast the accepted and processed set in round $r-1$ */}\\%

\lnl{line:f-end} \    {\bf end.} \\
\ \\
 {\underline {\bf In the Background:}} \\
\lnl{line:f-r1}   Execute  for each $\tri{r',p_i,\pi}\in\M$: \hfill\textit{/*   accepted message from $p_i$ for round $r'$   */} \\
\lnl{line:f-init-rec}    \tb {\bf if} $r'=1$ {\bf then start } $SM_i$ with input $\pi$;\hfill\textit{/* start a SM with the initial input */}\\
\lnl{line:f-init-rec2}     \tb {\bf if} $r'>1$ {\bf then } \\
\lnl{line:f-rec-m}  \tb\tb {\bf let} $M:=\{m_j\mid p_j\in\pi \mbox{ and } SM_j[r'-1] \mbox{ sends } m_j \mbox{ to } p_i\}$; \hfill\textit{/* the values $p_i$ should have received */}\\
\lnl{line:f-apply-m}  \tb\tb$SM_i[r']:=\F(M,SM_i[r'-1],r')$; \hfill\textit{/* apply  protocol $\F$ to determine the next state of $SM_i$ */}\\
\lnl{line:f-halted}  \tb\tb {\bf if } $SM_i$ outputs {\bf then} $O:=O\cup \{p_i\}$; \hfill\textit{/* Denote the outputting processor */}\\
\lnl{line:f-r3}    \tb $\M:=\M\setminus \tri{r',p_i,\pi}$;\\
\lnl{line:f-r4}     \tb $\bar\M:=\bar\M\cup \tri{r',p_i}$;\\
\lnl{line:f-add}       \tb   $accept[r']:=accept[r']\cup\{p_i\}$. \\
\caption[caption]{Simulating a \BIMObb system deterministic protocol  \\\hspace{\textwidth}\mbox{\ \hspace{0.83in}} in an asynchronous  system that satisfies  with $n>3t$.}\label{figure:full-state-sim}
\end{algorithm}

\begin{theorem}\label{thm:full}
Given  a deterministic protocol 
$P$ that is viewed as a function $\F(M_{r-1},S_p,r)$, for $r\ge 1$, 
in a \BIMObb system,
the protocol presented in \algref{figure:full-state-sim}  
simulates it  in a \ttbasyn, provided that $n>3t$.
\end{theorem}

\begin{corollary}\label{col:full}
For $n>3t$ and deterministic protocols, 
\ttbasyn and 
 \BIMObb system are equivalent.
\end{corollary}

\begin{theorem}\label{thm:both}
For $n>3t$  and deterministic protocols, 
\ttbasyn is  equivalent  to
\ttoasyn in a Cuckoo model (\ie in which the adversary only replaces  the inputs to some $t$ processors).
\end{theorem}

For the synchronous model with a fixed set of faults we obtain a similar result. To obtain the result we do not need the \ccore part and \algref{figure:state-sim} is enough.\footnote{One can somewhat optimize  the protocols in the synchronous case, but conceptually the simulation remains the same.}

\begin{theorem}\label{thm:sync}
For $n>3t$  and deterministic protocols.
A synchronous Byzantine fault system is equivalent to a 
synchronous system in which  
the adversary replaces  the inputs to some $t$ processors and outgoing messages from them  it may delete in the first round, and in any future round it may delete outgoing messages from some set of $t$ processors.
\end{theorem}

\section{Reliable Broadcast}\label{sec:RB}

In this section we present the traditional Reliable Broadcast protocol using the generalized faulty set convention.

\begin{definition}[Collection of Sets of Potentially Faulty Processors One of which can be Chosen by the Adversary]
The set $\B$, called the collection of \emph{bad-sets}, is closed under inclusion, i.e.
if $B \in \B$ and $B' \subseteq B$, then $B' \in \B$. 
Let $\G$ be $\{{\Pi \setminus B}~|~B \in \B\}$.
\beginsmall{itemize}
\item[B1:] $\B$ satisfies the Benign Fault predicate if: 
 any 2 potentially bad sets of processors $B_i,B_j\in\B$, satisfy $\Pi\not\subseteq B_i\cup B_j$. 
\item[B2:] $\B$ satisfies the Byzantine Fault predicate if: 
any 3 potentially bad sets $B_i,B_j,B_k\in\B$, satisfy  $\Pi\not\subseteq B_i\cup B_j\cup B_k$. 
\endsmall{itemize}
\end{definition}

Thus, the only difference between the Benign Fault predicate and the Byzantine Fault predicate is the number of sets that do not cover the entire set of processors.

\begin{corollary}\label{cor:bg-sets}
The definition of the potentially Faulty sets that satisfy the Byzantine Fault predicate implies:
\beginsmall{itemize}
\item[G1:] For every $G\in\G$, and any  $B_i,B_j\in\B$, it holds that $G\not\subseteq B_i\cup B_j$.
\item[G2:] For any 2 potentially good sets $G_i,G_j\in\G$, it holds that $G_i\cap G_j\subsetneq B$, $B\in\B$.
\endsmall{itemize}
\end{corollary}

\begin{proof}
Proving G1: If G1 does not hold then $\Pi\subseteq B_i\cup B_j\cup\{\Pi\setminus G\}$, a contradiction.\\
Proving G2: If G2 does not hold then $\Pi\subseteq \{\Pi\setminus G_i\}\cup\{\Pi\setminus G_j\}\cup\{G_i\cap G_j\}$, a contradiction, since all sets belong to $\B$.
\end{proof}

The Reliable Broadcast properties   are:
\beginsmall{enumerate}
\item[RB1] If the sender $s$ is correct when it invokes RB with a value $v$,  then every processor will eventually accept RB$(s)$ with the value $v$. 
\item[RB2] If any (correct) processor accepts RB$(s)$ with a value $v'$, then every processor will accept it with the same value.
\item[RB3] If sender $s$ did not invoke RB$(s)$ no processor will accept RB$(s)$.
\endsmall{enumerate}

Notice that in our model there is no need to use the word ``correct'' in stating RB2.

\begin{algorithm}[!ht]
\footnotesize
\SetNlSty{textbf}{}{.}
\setcounter{AlgoLine}{0}
\  \hfill\textit{/* executed by processor $p$ with  sender  $s$, invoked once per $k$ per sender */}\\

 \lnl{line:rb-def}  \     {\bf set}  $V:=\emptyset;$\hspace{2mm} \hfill\textit{/*  the set of $m_1$ and $m_2$  messages received,  each processor also sends  messages to itself */}\\

\ \\
\lnl{line:rb-step0}    {\bf Init:}  {\bf if}  $p=s $ {\bf then}  send $v_s$ to all;\mbox{\ }\hfill\textit{/* $s$ is the sender and $v_s$ the value it broadcasts */}\\

\ \\
\lnl{line:rb-exec1}    {\bf Upon receiving a protocol message:} \\
\lnl{line:rb-exec2}    \tb {\bf case}  received $v$ from $s$:  
send $m_1(v)$ to all;\hfill\textit{/* executed at most once per protocol invocation */}\\
\lnl{line:rb-exec2a}    \tb {\bf case}  received $m$   from $q$: 
 add $m$ to $V$;\hfill\textit{/* add this $m1$ or $m2$ message to $V$ */}\\
\lnl{line:rb-exec3}   \tb {\bf case}  $V$ contains $m_1(v)$ from  a set $G$, $G\in\G$  {\bf or} $m_2(v)$ from  a set $G'$, $G'\not\subseteq B$, for any $B\in\B$: \\
\due    send $m_2(v)$ to all;\hfill\textit{/* send at most once per protocol invocation */}\\

\ \\
\lnl{line:rb-exec4}   \tb {\bf case}  $V$ contains $m_2(v)$ from   a set $G\in\G$: \hfill\textit{/* process the sender's message */}\\
\lnl{line:rb-acpt} \due{\bf accept}  RB$(k,s)$ with value $v$.\hfill\textit{/*   accept message $v$ as the message sent by processor $s$  with index $k$  */} \\
\caption[caption]{ RB$(k,s)$: Reliable Broadcast \\\hspace{\textwidth}\mbox{\ \hspace{0.841in}} with asynchronous  system that satisfies  the Byzantine Fault predicate}\label{figure:rb}
\end{algorithm}

The Reliable Broadcast protocol, \algref{figure:rb}, is invoked per processor per an index $k$ per sending round and consists of 4 conceptual steps.   
Initially (step 1, \lref{line:rb-step0}) the sender of the current instance of the protocol sends its initial value to everyone.  
Thus, everyone should wait to receive the appropriate initial value. 
Due to asynchrony it may take time, but without faults, it would eventually arrive to everyone.  
Because of maliciousness, the message may not arrive to every processor.  Moreover, conflicting values might be sent to different processors.  The following steps intend to address exactly these difficulties.

If a processor receives an initial value (step 2, \lref{line:rb-exec2}) it notifies every processor by sending  $m_1(v)$ message.
Malicious behavior may cause different processors to send  $m_1$ messages for different values. 
Each processor sends at most a single  $m_1$ message per invocation of the protocol (per round).
A processor may receive several $m_1$ messages, even if it did not receive an initial value. 

In  step 3 (\lref{line:rb-exec3}), a processor that has received 
 identical copies of  $m_1$ messages (\ie for the same value) from a set 
 $G\in\G$, sends an $m_2$ message. 
Notice that if the original sender is correct, this will eventually happen at every correct processor.
Observe that by \corollaryref{cor:bg-sets} no two correct processors send  $m_2$ messages with conflicting values, since the protocol instructs a correct processor to send at most a single  $m_1$ message, and
no two correct processors will get  identical copies of $m_1$ messages for different values from different set in $\G$.
Notice that a processor may receive several  $m_2$ messages without receiving enough identical copies of $m_1$ messages from a set  $G\in\G$.
If it receives  $m_2$ identical messages from a set
$G'$, $G'\not\subseteq B$, for any $B\in\B$,
it knows that at least one correct processor have sent one, so it can also join in  by sending an  $m_2$ message (potentially skipping step 2 on the way).

To complete  the protocol (step 4, \lref{line:rb-exec4}) a processor  waits to receive   $m_2$ identical messages
from a set $G\in\G$. 
Once it receives that many identical $m_2$ messages, and
since any correct processor will either send the identical message or none,
it knows that eventually every processor will receive 
the $m_2$ message from at least  a set $G'$, $G'\not\subseteq B$, for any $B\in\B$.
Any such correct processor will send an  $m_2$ message to everyone else, which leads to everyone eventually receiving identical  $m_2$ messages from a set  $G\in\G$. 
Thus, if any correct processor will agree to accept the message, eventually every correct processor will accept and process all the prior messages and will accept the message.

Notice that the way \lref{line:rb-exec3} is presented it implies scanning an exponential number of sets - but it is clear that it is enough to consider maximal sets, and the complexity becomes the number of maximal sets.



\section{Related Work}

The attempts of reducing  Byzantine to  Benign have a long history. Back in 1988 Coan \cite{coan88} considered running an asynchronous protocol  written for the benign setting on a machine in which processors might fail in a Byzantine manner. To deal with inputs he assumed that the inputs satisfy some predicate, and then required the protocol to be written in a special form. This allowed him to check ``backward compatibility'' and discard incorrect messages. But, it was not clear whether any asynchronous protocol can be put in the form he required.

We, here, use the same idea as that of Coan, only that we do not require the special protocol form required by Coan. 
We check ``backward compatibility'' just by asking a processor to send the set of processors it read from. Since it is sent via Reliable Broadcast it allows the rest of the processors to delay accepting the message until they simulate locally all the processors claimed in the message and then derive the interpreted value themselves. In hindsight we believe Coan method can be tweaked to ours. Ours might be just an inductive full-information version of Coan ideas.

As for general results reducing the Asynchronous Byzantine to the Benign such a reduction was derived recently in
\cite{MTH14}, and rediscovered in \cite{IRS16}. But both reductions were explained only for \emph{colorless} tasks.
They might apply to all tasks, but we will speculate that the fact that they could not ascribe meaning to output of ``byzantine processors'' held them back. In fact, some of the author of \cite{MTH14} were involved in formulating the beautiful task of Multi-Dimensions Vector $\epsilon$-agreement \cite{MHVG15}. The task is motivated by the presence of Byzantine processors, but its computability is possible to analyze in the Benign model according to the Cuckoo model. Nevertheless its computability in \cite{MHVG15} is handled in the Byzantine setting with all the complications and clutter  it introduces.
Reference \cite{MTH14} is mute about that possibility, although we view it as a ``killer application'' to the benefit of the Cuckoo Model.

Thus the central contribution of our paper is conceptual.
By making all processors ``correct'' and allow them to join the computation once
the Byzantine attack ceases we are forced to think a complete analogy between
the Byzantine and the Benign modulo changed inputs.
The ``byzantine processors'' cluttered the view of researchers, hence for instance, we are the first to notice that the Asynchronous Adversary model of 2009 \cite{DFGT11}
applies verbatim to the Byzantine model. Missing the ``correct'' model might have cause researchers to miss results which are essentially in plain view.

In between Coan \cite{coan88} and \cite{MTH14} there were several attempts to simplify the $t$ resilient asynchronous Byzantine model  through ideas of simulating simpler models.  Attiya and Welch~\cite{DistComp} reduced the problem to Identical Byzantine.
The pioneering work of Bracha~\cite{Bracha:1984,Bracha:1987} was focused on improving the probabilistic protocol of Ben-Or~\cite{BenOr83} from $n/5$ to $n/3$ and in order to do so Bracha developed a basic tool to limit the power of the Byzantine adversary,  The simulation we introduce in the paper makes use of this tool as part of the building block we introduce.
Srikanth and Toueg~\cite{Srikanth87} considered simulating the power of a signature scheme to limit the Byzantine adversary, both in a synchronous system and an asynchronous one. Neiger and Toueg~\cite{neiger88} introduced direct simulations between models in order to solve consensus, but their simulations are limited to synchronous models.

\hide{
No attempt that we know of has been done to analyze the Adversary Correlated Faults in the Asynchronous Byzantine model. We again assume researchers had this blind-spot as a result of the ``byzantine processors.'' In this paper we just show the possibility of results derived in the benign setting to be imported to Byzantine. Thus we do not need to review the history of the development of the Adversary model.
}

Last but not least as mentioned in the introduction our motivating papers were \cite{AG15} and \cite{GG11,DFGK15}.
Reference \cite{GG11} implies that processors $t$-resiliently jointly can march any number (greater than $t$) of state machines forward with at most $t$ not progressing. In \cite{DFGK15} state-machines are considered to be read-write threads
(and the commands to be what a processors proposes a a ``read'' value based on its
internal simulation), hence a protocol. Hence it makes an execution a group effort - our main motivating observation as to how this might apply to the Byzantine.
The \cite{GG11} is a generalization of Schneider~\cite{FTState90} using a state machine  for implementing 
fault-tolerant services.  It generalizes \cite{FTState90} from when consensus is available to the case when set-consensus is available rather than consensus.
The $t$ resilience allows for $k+1$ set-consensus.

\section{Conclusions}

We presented a new view on the Byzantine model. 
In this view processors are not taken over by a malicious Adversary.
The Adversary takes over the outgoing communication attempted by some processors.
For all practical purposes the observable behavior is the same.
Yet theoretically it allows us to consider ``reviving'' a processor after an attack ceases
as throughout the execution it snooped on the system progress.
Thus, this raised the challenge as to whether there is a way to completely
reduce the Byzantine to the Benign with no qualifications.

This paper present a way to do this. It is not prefect.
We assume that faulty behavior will not be experienced further at some processor without output,
\emph{immediately} after a new output is delivered past $n-t-1$ outputs.
Perfection would mean to replace this assumption of ``immediately'' with ``eventually.''
We leave it as an open question whether this can be done
(e.g. by changing the assumption that the adversary can remove a message, to the one where the adversary can delay a message ad infinitum). Recovery necessitates that a processor trying revival will resend a message. We discarded various schemes 
of ``eventual'' as we could not bound the number of these extra resend messages.

Yet, with this somewhat flawed assumption, that nevertheless does not affects ``observable behavior,''
we are able to reduce the extreme of Asynchronous Byzantine, to the other extreme of the 
Synchronous Benign. What is left is to show that such reductions are true for other models e.g.
those that utilize objects in their communications. Indeed, implicitly such a reduction has 
been shown in \cite{G2016} 
for objects of set-consensus types, since there, objects are replaced by a restriction on runs.

It is now about making the reduction in this paper efficient, and finding ``killer applications,''
beyond \cite{MH13,VG13,MHVG15}. 
Another direction is turning a probabilistic asynchronous consensus algorithms,
into probabilistic algorithms where the faults are benign.

The emulations discussed in this paper assume deterministic protocol.  In order to extend the emulations above to randomized protocols, one can't just ask a processor to distribute together with its collected set its random choice for the current round. The reason is that the adversary can make the Byzantine processors to not draw the random bits from the expected distribution.  To deal with that all processors collectively can choose the random choices for each processor once it broadcasts its candidates' set.
In case $n>4t$ one can compute any probabilistic function using Asynchronous Multi-Party-Computation,  in the presence of up to $t$ Byzantine faults, see~\cite{Ben-Or:1993,BH2007}.  One can use that as a building block in our protocols. The question of doing that for $n>3t$ is an open problem.



\newpage
{
\bibliographystyle{plain}
\bibliography{bibliography}
}
\end{document}